\documentclass[arxiv]{stacks}

\title{Stack Semantics of Type Theory}
\author{Thierry Coquand\\Göteborgs Universitet \and Bassel Mannaa\\IT-Universitetet i København \and Fabian Ruch\\Göteborgs Universitet}
\date{April 2017}

\begin{document}
\renewcommand{\onlyinsubfile}[1]{}
\renewcommand{\notinsubfile}[1]{#1}

\maketitle

\begin{abstract}
We give a model of dependent type theory with one univalent universe and propositional truncation
interpreting a type as a {\em stack}, generalizing the groupoid model of type theory.
As an application, we show that countable choice cannot be proved in dependent type theory with one univalent
universe and propositional truncation.
\end{abstract}

\section{Introduction}

\label{sec:intro}
\onlyinsubfile{
  \title{Stack Semantics of Type Theory: Introduction}
  \maketitle
}

 The axiom of univalence \cite{2014arXiv1402.5556V,hottbook} can be seen as an extension  to dependent type theory
of the two axioms of extensionality for simple type theory as formulated by Church \cite{Church40}.
This extension is important since, using
universe and dependent sums, we get a formal system in which we can represent arbitrary
structures  (which we can not do in simple type theory) with elegant formal properties.
The goal of this paper is to contribute to
the meta-theory of such systems by showing that {\em Markov's principle} and {\em countable
choice} are not provable in dependent type theory extended with one univalent universe
and propositional truncation. For simple type theory such independence results can be
obtained by using {\em sheaf semantics}, respectively over Cantor space (for Markov's principle)
and open unit interval $(0,1)$ (for countable choice). There are however problems
with extending sheaf semantics to universes \cite{hsuniverses,xesheafuu}. In order to
address these issues
we use a suitable formulation of {\em stack semantics}, which, roughly speaking, replaces
{\em sets} by {\em groupoids}. The notion of stack was introduced in algebraic geometry \cite{EGA.i,giraud}
precisely in order to solve the same problems that one encounters when trying to extend sheaf
semantics to type-theoretic universes.
The compatibility condition for gluing local data is now formulated
in terms of isomorphisms instead of strict equalities.
In this sense, our model can also be seen as an extension of the groupoid
model of type theory \cite{MR1686862}. One needs to formulate some
strict functoriality conditions on the stack gluing operation, which seem necessary
to be able to get a model of the required equations of dependent type theory.

 We see this work as a first step
towards the proof of independence of countable choice from type theory with
a hierarchy of univalent universes and propositional truncation, which we hope to obtain by an
extension of our model to an $\infty$-stack version of cubical type theory \cite{CCHM}.

 The paper is organized as follows. We first present a slight variation of the groupoid model that we
find convenient for expressing the stack semantics. We then
explain how to represent propositional truncation in this setting, and how it can be used
to formulate countable choice. We then notice that, even in a constructive meta-logic where countable
choice fails, the axiom of countable choice does hold in this groupoid model. The groupoid model can
be refined rather directly over a Kripke structure, and we present then our notion of stacks over
a general topological space together with a
proof that we get a model of dependent type theory with one univalent universe and propositional
truncation. Instantiating our model to the case of Cantor space
and open unit interval $(0,1)$ we obtain the results that Markov's principle and countable choice cannot
be proved in dependent type theory with one univalent universe and propositional truncation.

\onlyinsubfile{
  \bibliography{bibliography}
}

\section{Type theory}

\label{sec:tt}
\onlyinsubfile{
  \title{Type theory}
  \maketitle
}

As in \cite{BCH}, we will use a generalized algebraic presentation of type theory that is
name-free and has explicit substitutions.
For instance, if we write $A\rightarrow B$ for $\Pi A (B\p)$ then we have
$\Gamma\vdash \lambda \q:A\rightarrow A$ since $\Gamma.A\vdash \q:A\p$.
The advantage of using such a presentation is that it makes it easier to check
the correctness of the model: Building such a model is reduced to defining
operations such that certain equations hold.
The main rules are presented in
figures \ref{fig:tt}, \ref{fig:tt-univ}, \ref{fig:tt-path} and \ref{fig:tt-sigma-nat-bool}.
We omit equivalence, congruence and substitution rules.
%Some context weakenings will also be left implicit.
The conversion rules assume appropriate typing premises.

\begin{figure}
\begin{mathpar}
\inferrule*{}{⊢ \emp}\and
\inferrule*{Γ ⊢ A}{⊢ Γ.A}\\
\inferrule*{⊢ Γ}{⊢ \id : Γ → Γ}\and
\inferrule*{⊢ τ : Θ → Δ \\ ⊢ σ : Δ → Γ}{στ : Θ → Γ}\and
%\inferrule*{⊢ Γ}{⊢ \bang : Γ → \emp}\\
\inferrule*{Γ ⊢ A \\ ⊢ σ : Δ → Γ}{Δ ⊢ Aσ}\and
\inferrule*{Γ ⊢ A}{Γ.A ⊢ \q : A\p}\and
\inferrule*{Γ ⊢ a : A \\ ⊢ σ : Δ → Γ}{Δ ⊢ aσ : Aσ}\\
A\id = A\and
A(στ) = (Aσ)τ\and
a\id = a\and
a(στ) = (aσ)τ\\
\inferrule*{Γ ⊢ A}{⊢ \p : Γ.A → Γ}\and
\inferrule*{Γ ⊢ A \\ ⊢ σ : Δ → Γ \\ Δ ⊢ a : Aσ}{⊢ (σ,a) : Δ → Γ.A}\\
\idσ = σ\and
σ\id = σ\and
σ(τυ) = (στ)υ\and
%\bang = σ\and
\p(σ,a) = σ\and
\q(σ,a) = a\and
(\pσ,\qσ) = σ\\
\inferrule*{Γ.A ⊢ B}{Γ ⊢ ΠAB}\and
\inferrule*{Γ.A ⊢ b : B}{Γ ⊢ \lam{b} : ΠAB}\and
\inferrule*{Γ ⊢ f : ΠAB \\ Γ ⊢ a : A}{Γ ⊢ \app{f}{a} : B\ir{a}}\\
\app{λb}{a} = b\ir{a}\and
λ\app{f\p}{\q} = f
\end{mathpar}
\caption{\label{fig:tt}Type theory}
\end{figure}

We write $\ir{a}$ for the substitution $(\id,a)$ and $\ir{a,b}$ for $(\ir{a},b)$.
%We write $\ir{a}$ and $\il{σ}$ for the substitutions $(\id,a)$ and $(σ,\q)$, respectively.

\begin{figure}
\begin{mathpar}
\inferrule*{Γ ⊢ A\ \textsf{small} \\ Γ.A ⊢ B\ \textsf{small}}{Γ ⊢ ΠAB\ \textsf{small}}\and
\inferrule*{Γ.A ⊢ B\ \textsf{discrete}}{Γ ⊢ ΠAB\ \textsf{discrete}}\\
\inferrule*{}{\Gamma ⊢ \Univ}\and
\inferrule*{Γ ⊢ A\ \textsf{small\ discrete}}{Γ ⊢ \In{A} : \Univ}\and
\inferrule*{Γ ⊢ a : \Univ}{Γ ⊢ \El{a}\ \textsf{small\ discrete}}\\
\El{\In{A}} = A\and
\In{\El{a}} = a\\
\inferrule*{Γ ⊢ A\ \textsf{small}}{Γ ⊢ A}\and
\inferrule*{Γ ⊢ A\ \textsf{discrete}}{Γ ⊢ A}
\end{mathpar}
\caption{\label{fig:tt-univ}Universe in type theory}
\end{figure}

\begin{figure}
\begin{mathpar}
\inferrule*{Γ ⊢ A \\ Γ ⊢ a : A \\ Γ ⊢ b : A}{Γ ⊢ \Path{A}{a}{b}\ \textsf{discrete}}\and
\inferrule*{Γ ⊢ A\ \textsf{small} \\ Γ ⊢ a : A \\ Γ ⊢ b : A}{Γ ⊢ \Path{A}{a}{b}\ \textsf{small}}\and
\inferrule*{Γ ⊢ A \\ Γ ⊢ a : A}{Γ ⊢ \refl{a} : \Path{A}{a}{a}}\and
\inferrule*{Γ ⊢ A \\ Γ.A.A\p.\Path{A\p\p}{\q\p}{\q} ⊢ C \\ Γ.A ⊢ c : C\ir{\q,\refl{\q}} \\ Γ ⊢ a : A \\ Γ ⊢ b : A \\ Γ ⊢ p : \Path{A}{a}{b}}{Γ ⊢ \J{c}{a}{b}{p} : C\ir{a,b,p}}\\
\J{c}{a}{a}{(\refl{a})} = c\ir{a}
\end{mathpar}
\caption{\label{fig:tt-path}Equality in type theory}
\end{figure}

%% \begin{figure}
%% \begin{mathpar}
%% \inferrule*{Γ ⊢ a : \Univ \\ Γ ⊢ b : \Univ}{Γ ⊢ \ua{a}{b} : \isEq{(\tp{a}{b})}}
%% \end{mathpar}
%% \caption{\label{fig:tt-ua}Univalence axiom}
%% \end{figure}

\begin{figure}
\begin{mathpar}
\inferrule*{Γ.A ⊢ B}{Γ ⊢ ΣAB}\and
\inferrule*{Γ ⊢ a : A \\ Γ ⊢ b : B\ir{a}}{Γ ⊢ \pair{a}{b} : ΣAB}\and
\inferrule*{Γ ⊢ p : ΣAB}{Γ ⊢ \fst{p} : A}\and
\inferrule*{Γ ⊢ p : ΣAB}{Γ ⊢ \snd{p} : B\ir{\fst{p}}}\\
\fst{\pair{a}{b}} = a\and
\snd{\pair{a}{b}} = b\and
\pair{\fst{p}}{\snd{p}} = p\\
\inferrule*{}{\Gamma ⊢ \Nat\ \textsf{small discrete}}\and
\inferrule*{}{\Gamma ⊢ \0 : \Nat}\and
\inferrule*{Γ ⊢ n : \Nat}{Γ ⊢ \suc n : \Nat}\and
\inferrule*{Γ.\Nat ⊢ C \\ Γ ⊢ c : C\ir{\0} \\ Γ.\Nat.C ⊢ d : C\ir{\suc \q}\p \\ Γ ⊢ n : \Nat}{Γ ⊢ \natrec{c\, d\,n} : C\ir{n}}\\
\natrec{c\, d \,\0} = c\and
\natrec{c\,d\,(\suc n)} = d\ir{n,\natrec{c\,d\,n}}\\
\inferrule*{}{\Gamma ⊢ \Bool\ \textsf{small discrete}}\and
\inferrule*{}{\Gamma ⊢ \0 : \Bool}\and
\inferrule*{}{\Gamma ⊢ \1 : \Bool}\and
\inferrule*{Γ.\Bool ⊢ C \\ Γ ⊢ c : C\ir{\0} \\ Γ ⊢ d : C\ir{\1} \\ Γ ⊢ b : \Bool}{Γ ⊢ \boolrec{c\,d\,b} : C\ir{b}}\\
\boolrec{c\,d\,\0} = c\and
\boolrec{c\,d\,\1} = d
\end{mathpar}
\caption{\label{fig:tt-sigma-nat-bool}Dependent sum, natural numbers and Booleans in type theory}
\end{figure}

\onlyinsubfile{
  \bibliography{bibliography}
}

\section{Groupoid model}

\label{sec:gpd}
\onlyinsubfile{
  \title{Groupoid model of type theory}
  \maketitle
}

 In this section, we review the {\em groupoid model} of \cite{MR1686862}, with a slightly
different presentation inspired from \cite{Fabian:Thesis:2015}.
We work in a set theory with a Grothendieck universe ${\cal U}$ (or a suitable constructive
version of it if we work in a constructive set theory such as CZF \cite{Crosilla200233}).

 A {\em groupoid} is given by
a set $\Gamma$ of objects and for each $\rho,\rho'\in \Gamma$ a
set $\Gamma(\rho,\rho')$ of paths/iso\-mor\-phisms along with a composition operation
$\alpha\cdot\alpha'$ in $\Gamma(\rho,\rho'')$ for $\alpha$ in $\Gamma(\rho,\rho')$
and $\alpha'$ in $\Gamma(\rho',\rho'')$ and a unit element $\id_{\rho}$ in $\Gamma(\rho,\rho)$ and
an inverse operation
$\inv{\alpha}$ in $\Gamma(\rho',\rho)$ satisfying the usual unit, inverse and associativity laws.
We may write $\alpha:\rho \pto \rho'$ for $\alpha$ in $\Gamma(\rho,\rho')$.

 A map $\sigma:\Delta\rightarrow\Gamma$ between two groupoids $\Delta$ and $\Gamma$ is given by a
set-theoretic map $\sigma\, \nu$ in $\Gamma$ for $\nu$ in $\Delta$ and a map
$\sigma\, \beta$ in $\Gamma(\sigma\, \nu,\sigma\, \nu')$ for $\beta$ in $\Delta(\nu,\nu')$ which commutes
with unit, inverse and composition.

 A family $A$ of groupoids indexed over a groupoid $\Gamma$, written $\Gamma\vdash A$,
is given by a family
of sets $A\rho$ for each $\rho$ in $\Gamma$ and sets $A\alpha (u,u')$ for each $\alpha$ in
$\Gamma(\rho,\rho')$ and $u\in A\rho$ and $u'\in A{\rho'}$. We may write $\omega:u \pto_{\alpha} u'$
for $\omega$ element of $A\alpha(u,u')$ and we may omit the subscript $\alpha$ if it is clear
from the context.
We also have unit $\id_u : u \pto_{\id_{\rho}} u$ and inverse $\inv{\omega}:u' \pto_{\inv{\alpha}} u$
and composition $\omega\cdot\omega':u \pto_{\alpha\cdot\alpha'} u''$ also satisfying
the unit, inverse and associativity laws.
We furthermore should have a {\em path lifting structure}, which is given by two operations
$u\alpha$ in $A\rho'$ and $u\uparrow\alpha:u \pto_{\alpha} u\alpha$ for $u$ in $A\rho$
and $\alpha:\rho \pto \rho'$ satisfying the laws
\begin{gather*}
u\id_{\rho} = u \qquad (u\alpha)\alpha' = u(\alpha\cdot\alpha') \qquad
u\uparrow \id_{\rho} = \id_u \qquad (u\uparrow\alpha)\cdot (u\alpha\uparrow\alpha') = u\uparrow (\alpha\cdot\alpha')
\end{gather*}
 We see that $u\uparrow\alpha$ ``lifts'' the path $\alpha:\rho \pto \rho'$ given an initial point
$u$ in $A\rho$.

 Each $A\rho$ has a canonical groupoid structure, defining $A\rho(u,u')$ to be $A\id_{\rho}(u,u')$.
If $\alpha:\rho \pto \rho'$ we can define a groupoid map $A\rho\rightarrow A\rho'$ using the lifting
operation. We thereby recover the groupoid model as defined in \cite{MR1686862}.

 If $\sigma:\Delta\rightarrow\Gamma$ and $\Gamma\vdash A$ we define $\Delta\vdash A\sigma$
by composition: $(A\sigma)\nu$ is $A(\sigma\,\nu)$ and $(A\sigma)\beta(v,v')$ is
$A(\sigma\,\beta)(v,v')$.

 A {\em section} $\Gamma\vdash a:A$ is given by a family of objects $a\rho$ in $A\rho$ together
with a family of paths $a\alpha:a\rho \pto_{\alpha} a\rho'$ satisfying the laws
$a\id_{\rho} = \id_{a\rho}$ and $a(\alpha\cdot\alpha') = a\alpha\cdot a\alpha'$.

 If $\Gamma\vdash A$, we define a new groupoid $\Gamma.A$: An object $(\rho,u)$ in $\Gamma.A$ is a pair
with $\rho$ in $\Gamma$ and $u$ in $A\rho$ and a path $(\alpha,\omega):(\rho,u) \pto (\rho',u')$
is a pair $\alpha:\rho \pto \rho'$ and $\omega:u \pto_{\alpha} u'$. We then
have $\p:\Gamma.A\rightarrow \Gamma$ defined by $\p (\rho,u) = \rho$ and
$\p (\alpha,\omega) = \alpha$ and the section $\Gamma.A\vdash\q : A\p$ defined
by $\q (\rho,u) = u$ and $\q (\alpha,\omega) = \omega$.

 We say that a family $\Gamma\vdash A$ is {\em small} if each set $A\rho$ and $A\alpha(u,u')$ is
in the given
Grothendieck universe ${\cal U}$. We say that this family is {\em discrete} if the lifting is
{\em uniquely} determined: Given $u$ in $A\rho$ and $\alpha:\rho \pto \rho'$ there is a unique
$u'$ in $A\rho'$ such that $A\alpha(u,u')$ is inhabited and this set is a singleton in this case.
This notion of discrete family can be characterized in terms of the
common definition of \emph{discrete groupoid}, which says that a
groupoid is discrete if the only paths are units.

\begin{lemma}
$Γ ⊢ A$ is discrete if and only if each groupoid $Aρ$, $ρ ∈ Γ$, is
discrete.

\begin{proof}
Assume $Γ ⊢ A$ to be a discrete family and let $ω ∈ A\id_ρ(u,u')$ be
an arbitrary path. We immediately have $u' = u$ and $ω = \id_u$ by
discreteness of $Γ ⊢ A$ and $\id_u ∈ A\id_ρ(u,u)$.

For each $ρ ∈ Γ$, assume $Aρ$ to be a discrete groupoid and let $ω' ∈
Aα(u,u')$, $ω'' ∈ Aα(u,u'')$ be two arbitrary paths over some $α ∈
Γ(ρ,ρ')$. Then, $ω''$ can be expressed as the composite of $ω'$ and
$\inv{ω'} · ω'' ∈ A\id_{ρ'}(u',u'')$. The discreteness of $Aρ'$
forces $\inv{ω'} · ω''$ to be a unit path so that $u'' = u'$ and $ω''
= ω'$.
\end{proof}
\end{lemma}

 We define $\UU$ to be the following groupoid: An object $X$ in $\UU$ is exactly an element of the given
Grothendieck universe ${\cal U}$, and an element of $\UU(X,X')$ is a bijection between $X$ and $X'$.
We can then define the small and discrete family $\UU\vdash \El\!$ by taking $\El X$ to be the set $X$
and $u \pto_{\alpha} u'$ to be the subsingleton set $\{0\, |\, u' = \alpha u\}$, that is
$u \pto_{\alpha} u'$ is inhabited and is the singleton $\{0\}$ exactly when $u' = \alpha u$.

\begin{proposition}
The family $\UU\vdash \El\!$ is a {\em universal} small and discrete family: If $\Gamma\vdash A$ is
small and discrete, then there exists a {\em unique} map $|A|:\Gamma\rightarrow \UU$ such that
$\El |A| = A$ (with {\em strict} equality).
\end{proposition}

 For $\Gamma\vdash A$ and $\Gamma.A\vdash B$ we define $\Gamma\vdash\Pi A B$ by
taking $(\Pi A B)\rho$ to be the set of functions $c\, u$ in $B(\rho,u)$ and
$c\, \omega$ in $B(\id_{\rho},\omega)(c\, u,c\, u')$ commuting with unit and composition,
and $(\Pi A B)\alpha(c,c')$ to be the set of functions $\gamma\, \omega:c\, u \pto_{(\alpha,\omega)} c'\, u'$
such that $(\gamma\,\omega_0)\cdot (c'\,\beta') = (c\,\beta)\cdot (\gamma\,\omega_1)$ if
$\beta:u_0 \pto_{{\rho}} u_1$ and
$\beta':u'_0 \pto_{{\rho'}} u'_1$ and
$\omega_0:u_0 \pto_{\alpha} u'_0$ and
$\omega_1:u_1 \pto_{\alpha} u'_1$. There is then \cite{MR1686862,Fabian:Thesis:2015} a canonical way to define
a composition operation (we need the path lifting structure for $\Gamma\vdash A$)
and path lifting structure for $\Gamma\vdash\Pi A B$.

\begin{proposition}
If $\Gamma.A\vdash B$ is discrete, then so is $\Gamma\vdash \Pi A B$.

\begin{proof}
In order to show that $Γ ⊢ ΠAB$ is a discrete family, it suffices to
show that $(ΠAB)ρ$ is a discrete groupoid for each $ρ ∈ Γ$. Assume
$Γ.A ⊢ B$ to be a discrete family and let $γ ∈ (ΠAB)\id_ρ(c,c')$ be
an arbitrary path. In particular, $Bu$ is a discrete groupoid forcing
$γ\,\id_{u} ∈ B\id_{u}(c\,u,c'\,u)$ to be a unit path for each $u ∈
Aρ$ so that $c'\,u = c\,u$ for all $u ∈ Aρ$. The discreteness of $Γ.A
⊢ B$ also forces $c'\,ω = c\,ω$ and $γ\,ω = \id_c\,ω$ in
$Bω(c\,u',c\,u'')$ for all $ω ∈ A\id_ρ(u',u'')$, which concludes $c'
= c$ and $γ = \id_c$.
\end{proof}
\end{proposition}

 If $\Gamma\vdash A$ and $\Gamma\vdash a_0:A$ and $\Gamma\vdash a_1:A$
we define the {\em discrete} family $\Gamma\vdash \Path{A}{a_0}{a_1}$.
We take $(\Path{A}{a_0}{a_1})\rho$ for $ρ ∈ Γ$ to be the set
$A\id_{\rho}(a_0\rho,a_1\rho)$ and $(\Path{A}{a_0}{a_1})\alpha(\omega,\omega')$ for
$\alpha:\rho \pto \rho'$ to be the subsingleton
$\{0\, |\, \omega\cdot a_1\alpha = a_0\alpha\cdot \omega'\}$.

 It is then possible \cite{MR1686862,Fabian:Thesis:2015}
to check that this defines a model of type theory as presented by the rules of
figures \ref{fig:tt}, \ref{fig:tt-univ}, \ref{fig:tt-path} and \ref{fig:tt-sigma-nat-bool}.

\subsection{Propositional truncation}

 We say that a groupoid is a {\em proposition} if and only if there exists
exactly one path between two objects. So $\Gamma$ is a proposition if and only if
each set $\Gamma(\rho,\rho')$ is a singleton.
More generally, we say that a family $Γ ⊢ A$ is a proposition if each
set $Aα(u,u')$ is a singleton.

\begin{lemma}
$Γ ⊢ A$ is a proposition if and only if each groupoid $Aρ$, $ρ ∈ Γ$,
is a proposition.

\begin{proof}
It is clear that each $Aρ$ is a proposition if the whole family $Γ ⊢
A$ is a proposition.

Assume now each $Aρ$ to be a proposition and let $α ∈ Γ(ρ,ρ')$ as
well as $u ∈ Aρ$, $u' ∈ Aρ'$. Then, the set $Aα(u,u')$ is inhabited
by the composite $(u↑α) · p_{uα,u'}$ of the lifting of $u$ over $α$
with the unique path between $uα$ and $u'$ in $Aρ'$. Furthermore, for
any two paths $ω, ω' ∈ Aα(u,u')$ the composite $\inv{ω} · ω'$ is
forced to be the unit path at $u'$ so that $ω' = ω · \inv{ω} · ω' =
ω$.
\end{proof}
\end{lemma}

 We define as usual (where names are used for readability)
\begin{gather*}
\isProp{A} = \Pi (x_0\, x_1:A) \Path{A}{x_0}{x_1}
\end{gather*}

\begin{proposition}
If $\Gamma\vdash A$, then there exists a section $\Gamma\vdash p:\isProp{A}$ if and only if
each groupoid $A\rho$, $\rho$ in $\Gamma$, is a proposition.

\begin{proof}
It is enough to show that there exists a family of paths $p_{ρ,u,u'}
∈ A\id_ρ(u,u')$, $u, u' ∈ Aρ$, $ρ ∈ Γ$, satisfying $p_{ρ,u,u'} · ω' =
ω · p_{ρ',v,v'}$ for all $ω ∈ Aα(u,v)$ and $ω' ∈ Aα(u',v')$, $u, u' ∈
Aρ$, $v, v' ∈ Aρ'$, $α ∈ Γ(ρ,ρ')$, $ρ, ρ' ∈ Γ$, if and only if each
groupoid $Aρ$, $ρ ∈ Γ$, is a proposition.

Assume such a family $p$ and let $ρ ∈ Γ$, $u, u' ∈ Aρ$, then
$A\id_ρ(u,u')$ is inhabited by the composite $p_{ρ,u,u'} ·
\inv{(p_{ρ,u',u'})}$ and, moreover, any other path $ω ∈ A\id_ρ(u,u')$
satisfies $p_{ρ,u,u'} · \id_{u'} = ω · p_{ρ,u',u'}$ so that
$A\id_ρ(u,u')$ is indeed a singleton.

In the opposite direction, we can actually assume the whole family $Γ
⊢ A$ to be a proposition. Then, defining $p_{ρ,u,u'}$ to be the
unique path from $u$ to $u'$ satisfies $p_{ρ,u,u'} · ω' = ω ·
p_{ρ',v,v'}$ because there exists exactly one path from $u$ to $v'$
over $α$.
\end{proof}
\end{proposition}

 For $\Gamma\vdash A$ we define $\Gamma\vdash \norm{A}$ as follows. For each $\rho$ in $\Gamma$
 we take $\norm{A}\rho = A\rho$, and for each $\alpha$ in $\Gamma(\rho,\rho')$,
$u$ in $A\rho$ and $u'$ in $A\rho'$ we take $\norm{A}\alpha(u,u')$ to be a fixed singleton $\{0\}$.
We then have sections of $\Gamma\vdash \isProp{\norm{A}}$ and $\Gamma\vdash A\rightarrow\norm{A}$,
and given sections of $\Gamma\vdash \isProp{B}$ and $\Gamma\vdash A\rightarrow B$
there is a section of $\Gamma\vdash\norm{A}\rightarrow B$. In this way, we get a model
of the {\em propositional truncation} operation.

\subsection{Countable choice}

 The statement of {\em countable choice} can be formulated as the type \cite{hottbook}
\begin{gather*}
\CC = \Pi (A:\NN\rightarrow\UU) (\Pi (n:\NN) \norm{\El (A\, n)})\rightarrow \norm{\Pi (n:\NN) \El (A\, n)}
\end{gather*}

 Notice that we can develop the groupoid model in a constructive meta-theory where countable
choice may or may not hold.

\begin{theorem}
The statement $\CC$ is valid in the groupoid model (even if countable choice does not hold in the meta-theory).
\end{theorem}
\begin{proof}
It is enough to define $c\, A\, f = f$ and $c\, \alpha\, \omega = 0$ to get $\emp \vdash c:\CC$.
\end{proof}

\onlyinsubfile{
  \bibliography{bibliography}
}

\section{Stack model}

\label{sec:stack}
\onlyinsubfile{
  \title{Stack model of type theory}
  \maketitle
}

\subsection{Groupoid-valued presheaf model}

 We suppose given a poset with elements $U, V, W, X,…$.
The groupoid model extends directly as a groupoid-valued presheaf model over this poset.
A \emph{context} is now a family of groupoids $\Gamma(U)$ indexed by elements of the given poset
such that objects $\rho$ and paths $\alpha$ in
$Γ(U)$ can be restricted to $\rho|V$ and $\alpha|V$ in $Γ(V)$ if $V\subseteq U$ such that
the restriction operation defines a groupoid map $\Gamma(U)\rightarrow\Gamma(V)$ which
is the identity map for $V = U$ and the composite of $Γ(X) → Γ(V)$
and $Γ(U) → Γ(X)$ for $V ⊆ X ⊆ U$.

 For a given context $\Gamma$, we define then what is a family $\Gamma\vdash A$.
It is given by a family of sets $A\rho$ for
each $U$ and $\rho$ in $\Gamma(U)$ together with a restriction $u|V$ in $A(\rho|V)$
for $u$ in $A\rho$ satisfying $u|U = u$ and $(u|X)|V = u|V$, as well as a family of sets $A\alpha(u,u')$ for
each $α : ρ \pto ρ'$ in $Γ(U)$, $u$ in $Aρ$ and $u'$ in $Aρ'$ together with a restriction
$\omega|V$ in $A(\alpha|V)(u|V,u'|V)$ for $ω$ in $Aα(u,u')$ satisfying $ω|U = ω$ and $(ω|X)|V = ω|V$. In
particular, we require the restriction operation on the sets $Aα(u,u')$ to commute with unit and composition.
Such a family is called \emph{small} if
the sets $Aρ$ and $Aα(u,u')$ are elements in
the Grothendieck universe $\mathcal{U}$, and it is called a
\emph{proposition} if the canonical groupoid structure on each $Aρ$
defines a proposition, or, equivalently, if each set $Aα(u,u')$ is a
singleton.
Furthermore, we should have a lifting operation $u\uparrow\alpha$ with
the law $(u\uparrow\alpha)|V = (u|V)\uparrow (\alpha|V)$.
A family is called \emph{discrete} if
the liftings $u↑α$ are uniquely determined:
Given $U$ and $\rho$ in $\Gamma(U)$ and
 given $u$ in $A\rho$ and $\alpha:\rho \pto \rho'$ there is a unique
$u'$ in $A\rho'$ such that $A\alpha(u,u')$ is inhabited, and this set is a singleton in this case.

%In particular, the families of groupoids $Γ(U) ⊢ A(U)$ are discrete
%in this case.

 We can extend the groupoid model to this setting.

 An element $c$ of $(\Pi A B)\rho$
for $\rho$ in $\Gamma(U)$ is a function $c\, u$ in $B(\rho|V,u)$ for $V\subseteq U$ and
$u$ in $A(\rho|V)$ and $c\, \omega$ in $B(\id_{\rho|V},\omega)(c\, u,c\, u')$ for $\omega$ in $A\id_{\rho|V}(u,u')$
commuting with unit and composition
such that $(c\, a)|W = c\, (a|W)$ and $(c\, \omega)|W = c\, (\omega|W)$ if $W\subseteq V\subseteq U$.

%% Objects $u, u'$ in $\Univ(U)$ are
%%  small presheaves on the downward closure of $U$. A path $ω$ between $u$ and $u'$ is a
%% presheaf isomorphism. The canonical restriction map to an element $V ⊆
%% U$ is given by considering $u, u'$ and $ω$ only on the elements $W ⊆
%% V$.

% We can then define $\Univ\vdash \El\!$ which is universal as small and discrete family.

An element in $(\Sigma A B)\rho$ for $\rho \in \Gamma(U)$ is a pair $(a, b)$ where $a \in A\rho$ and $b\in B(\rho, a)$ with restrictions $(a, b)|V = (a|V, b|V)$. Paths in $(\Sigma A B)\alpha ((a, b), (a', b'))$, where $\alpha:\rho \pto \rho'$ are pairs $(\omega, \mu)$ where $\omega: a \pto_\alpha a'$ and $\mu: b \pto_{(\alpha, \omega)} b'$ with restrictions $(\omega, \mu)|V = (\omega|V, \mu|V)$.

Given sections $a_0$ and $a_1$ of $A$, an element in
$(\Path{A}{a_0}{a_1})ρ$ for $ρ ∈ Γ(U)$ is a path $ω : a_0ρ \pto a_1ρ$
with restrictions as in $A$. For every element $ω$ and path $α : ρ
\pto ρ'$ there is a unique path from $ω$ over $α$ going to
$\inv{(a_0α)} · ω · a_1α : a_0ρ' \pto a_1ρ'$.

\subsection{Stack structure}

 We assume given a topological space with a notion of {\em basic open} closed under nonempty intersection
and a notion
of {\em covering} of a given basic open by a family of basic opens. We consider only
coverings $(U_i)_{i ∈ I}$ of some basic open $U$ where the set of indices $I$ is {\em small}.
To simplify the presentation
we assume that each basic open set is {\em nonempty}.
We write $U_{ij}$ for $U_i\cap U_j$ and
$U_{ijk}$ for $U_i\cap U_j\cap U_k$ when they are nonempty.

 Since basic opens form a poset, we can consider the notion of type family
over this poset as defined in the previous subsection.

In the following we will define what is a {\em stack structure} on a type family.

 We recall that a {\em sheaf} $F$ is given by a presheaf, i.e.\ a family of {\em sets} $F(U)$ with restriction
maps $u|V$ in $F(V)$ for $V\subseteq U$ such that $u|U = u$ and $(u|V)|W = u|W$ if
$W\subseteq V\subseteq U$, which satisfies the condition that if we have a covering $(U_i)_{i ∈ I}$
of $U$ and a family of compatible elements $u_i$ in $F(U_i)$ (i.e.\ $u_i|U_{ij} = u_j|U_{ij}$)
then there exists a unique $u$ in $F(U)$ such that $u|U_i = u_i$ for all $i$.

 A type family $Γ ⊢ A$ is called a \emph{prestack} if it satisfies
the following sheaf condition on paths: If $α : ρ \pto ρ'$ is in $\Gamma(U)$,
$u$ and $u'$ are in $Aρ$ and $Aρ'$ respectively
and we have a family of paths $\omega_i : u|U_i \pto_{α|U_i} u'|U_i$ which is compatible
(that is $\omega_i|U_{ij} = \omega_j|U_{ij}$), then
we have a {\em unique} path $\omega : u \pto_α u'$ such that $\omega|U_i = \omega_i$
for all $i$.

For each basic open $U$ and $\rho$ in $\Gamma(U)$ we define what is
the set of {\em descent data} $D(A)\rho$.
A \emph{descent datum} is given by a covering $(U_i)_{i ∈ I}$ of $U$ and
a family of objects $u_i ∈ A(\rho|U_i)$
with paths $\varphi_{ij} : u_i|U_{ij} \pto_{\rho|U_{ij}} u_j|U_{ij}$, when $U_i$ meets $U_j$,
satisfying the \emph{cocycle} conditions\footnote{The first condition is not logically necessary.}
\begin{gather*}
\varphi_{ii} = \id_{u_i} \qquad \varphi_{ij}|U_{ijk} · \varphi_{jk}|U_{ijk} = \varphi_{ik}|U_{ijk}
\end{gather*}
This forms a set since the index set is restricted to be small (otherwise this might be a proper
class in general).

 If $d = (u_i,\varphi_{ij})$ is an element of $D(A)\rho$ and $V\subseteq U$
we define its restriction $d|V$, element of $D(A)\rho|V$, which is the family
$(u_i|V\cap U_i,\varphi_{ij}|V\cap U_{ij})$ restricted to indices $i$ such that $V$ meets
$U_i$. A {\em gluing operation} $\glue d = (u,\varphi_i)$ gives an element
$u$ in $A\rho$ together with paths $\varphi_i : u|U_i \pto u_i$
such that $\varphi_i|U_{ij}\cdot\varphi_{ij} = \varphi_j|U_{ij}$
{\em and} satisfies the law $(\glue d)|V = \glue (d|V)$, that is
$\glue (d|V)$ should be $(u|V,\varphi_i|V\cap U_i)$ where we restrict the family to indices
$i$ such that $V$ meets $U_i$.
 This functoriality property will be crucial for checking
that we do get a model of type theory with dependent product.

 A {\em stack structure} on a prestack $\Gamma\vdash A$ is given by a gluing
operation\footnote{Notice that we shall not require the context $\Gamma$ to be a prestack or have a stack structure.}.

Consider a prestack $\Gamma \vdash A$ and descent datum $d=(u_i,\varphi_{ij}) \in D(A)\rho$ with $\glue d = (u,\varphi_i)$ as above. Let $v \in A\rho$ with paths $\vartheta_i:v|U_i \pto u_i$ satisfying $\vartheta_i|U_{ij} · \varphi_{ij} = \vartheta_j | U_{ij}$. We remark that while it is not necessarily true that $u=v$, the prestack condition implies that we have a path $v \pto u$.

Note that while it is sufficient to define the notion of sheaf as a {\em property} because of the uniqueness part of the sheaf condition, it is crucial that our notion of stack is in general a {\em structure}, i.e.~given with an explicit operation fixing a particular choice of glue.

A stack is not the same as a groupoid object in the sheaf topos. A prime example of a stack whose presheaf of objects is not a sheaf is the universe of sheaves: If we define
$F(U)$ to be the collection of small sheaves over $U$ then there is a natural restriction
operation $F(U)\rightarrow F(V)$ for $V\subseteq U$, and
one can check that the gluing of a compatible family of elements is not unique up to strict
equality in general (but it is unique up to isomorphism).
Notice that if we try to define the stack structure
using global choice as in \cite[3.3.1, page 28]{EGA.i}
then the functoriality condition $(\glue d)|V = \glue (d|V)$ will not hold. There is however
a more canonical definition of gluing which satisfies this condition, which will provide the interpretation
of a univalent universe.

There is also a simple example of a prestack that is
not a stack but whose presheaf of objects is a sheaf. Consider the
topological space given by basic opens $U_1,U_2,U_{12}$ with $U_1 ∧
U_2 = U_{12}$ and the groupoid-valued presheaf $G$ given by the
propositions on the sets $G(1) = ∅$, $G(U_1) = \{x_1\}$, $G(U_2) =
\{x_2\}$ and $G(U_{12}) = \{x_1,x_2\}$. There are no matching families of
objects or morphisms in $G$ so that both the presheaf of objects and
the presheaf of morphisms trivially satisfy the sheaf property.
However, the descent datum given by $d_1 = x_1$, $d_2 = x_2$ and $d_{12}$
the unique path between $x_1|U_{12}$ and $x_2|U_{12}$ cannot have a glue
because $G(1)$ is empty.

Taking as objects the subset $D(A)(ρ,C) ⊆ D(A)ρ$ of descent data on
a covering $C = (U_i)_{i ∈ I}$ of $U$ and a path between two descent data
$(u_i,φ_{ij})$ and $(v_i,ψ_{ij})$ to be a family of paths $ω_i : u_i
\pto v_i$ satisfying $ω_i · ψ_{ij} = φ_{ij} · ω_j$ has a natural
groupoid structure. Moreover, the canonical restriction from $Aρ$ to
$D(A)ρ$ extends to a functor from the canonical groupoid structure on
$Aρ$ to $D(A)(ρ,C)$. The prestack condition for $A$ then says that
this functor is fully faithful and a gluing operation witnesses that
it is essentially surjective. If $A$ is a stack, then the canonical functor $Aρ \rightarrow D(A)(ρ,C)$ is
an equivalence of groupoids.

\subsection{Dependent product}

The collection of types with a stack structure is closed under
dependent product.

\begin{theorem}
If $Γ.A ⊢ B$ has a stack structure then $Γ ⊢ ΠAB$ has a stack structure.
\end{theorem}

\begin{proof}
Let $(u_i,φ_{ij}) ∈ D(ΠAB)ρ$ be a descent datum on a covering $(U_i)_{i ∈ I}$
of $U$. We construct a glue $(u,φ_i)$ that commutes with restriction.

Given $x, x' ∈ A(\rho|V)$ and $ν : x \pto x'$ on $V ⊆ U$, we
construct $(u\,x,φ_i\,x)$ as the glue of
$d_x = (u_i\,x,φ_{ij}\,x)$ and $u\,ν : u\,x \pto_ν u\,x'$ as the
unique path matching $u\,x \pto u_i\,x \pto_{ν|V∩U_i}
u_i\,x' \pto u\,x'$ given by the composite of $φ_i\,x$, $u_i\,ν$ and the inverse of
$φ_i\,x'$ on $V∩U_i$. If in particular $V ⊆ U_i$, then this completely
determines $φ_i : u|U_i \pto u_i$. The uniqueness of $u\,ν$ is needed to
show that $u$ respects units and composites as well as restriction of
paths. For $u$ to also respect restriction of objects we need
the fact that $(\glue d_x)|W = \glue d_x|W = \glue d_{x|W}$ for $W ⊆ V$.

Let $ω_i : u|U_i \pto_{α|U_i} u'|U_i$ be a matching family of paths.
We show that there is a unique glue $ω : u \pto_α u'$. It is
uniquely determined by the glues $ω\,ν : u\,x \pto_{(α|V,ν)} u'\,x'$
of $ω_i\,ν : u\,x \pto_{(α|V∩U_i,ν|V∩U_i)} u'\,x'$ for $ν : x \pto_{α|V} x'$.
In particular, $ω\,ν = ω_i\,ν$ if $V ⊆ U_i$. Again, the uniqueness of
$ω\,ν$ lets us show that $ω$ respects composites and restrictions.
\end{proof}

\subsection{Universe of sheaves}

 We define $\Univ(V)$ to be the collection of all {\em small} sheaves over $V$.
There is a natural restriction operation $\Univ(V) \rightarrow \Univ(W)$ if $W\subseteq V$.

\begin{theorem}
$\Univ$ has a stack structure.
\end{theorem}

\begin{proof}
Let
$F_i ∈ \Univ(U_i)$ with $φ_{ij} : F_i|U_{ij} \pto F_j|U_{ij}$
be a descent datum on a cover $(U_i)_{i ∈ I}$ of $U$. We construct a glue $F ∈
\Univ(U)$ and $φ_i : F|U_i \pto F_i$. We define $F(V)$ for $V ⊆ U$ as
the set of families $(x_i)_i$ where
$x_i ∈ F_i(V\cap U_i)$ and $φ_{ij}(x_i) = x_j$.
Furthermore, we define $F(V) → F(W)$ for $W ⊆ V$ component-wise by
the restriction $F_i(V\cap U_i) → F_i(W\cap U_i)$ and $φ_i$ by the projection to
the $i$-th component. For $φ_i$ to be an isomorphism we need the fact
$φ_{ii} = \id$ and $φ_{ij} · φ_{jk} = φ_{ik}$.

We claim that the presheaf $F$ satisfies the sheaf property.
Indeed, let $v_k ∈ F(V_k)$ be a matching family for $F$ on a cover $(V_k)_{k ∈ K}$ of
$V$. The $i$-th components of $v_k$ are a matching family for $F_i$
on the induced cover $(V_k\cap U_i)_{k ∈ K}$ of $V\cap U_i$
and the gluing operation $D(F_i)(V_k\cap U_i) → F_i(V\cap U_i)$ of a discrete stack is a bijection
so that we obtain a glue $v
∈ F(V)$ of $v_k$ by gluing component-wise. This glue is unique
because it is component-wise unique.

Let now $ω_i : G|U_i \pto H|U_i$, $G, H ∈ \Univ(U)$
be a matching family of paths on a cover $(U_i)_{i ∈ I}$ of $U$. For $x ∈ G(V)$, $V ⊆ U$ the family $ω_i\,x$ in
$H(V∩U_i)$ is compatible because $ω_i = ω_j$ on $V∩U_{ij}$. We define
$ω\,x$ to be the unique glue in $H(V)$ such that $(ω\,x)|V∩U_i =
ω_i\,x$. The uniqueness of glues allows us to verify that $ω$
respects restriction and that the such defined $ω$ is the unique path
that agrees with $ω_i$ on $U_i$.
\end{proof}

 We define $\Univ\vdash \El\!$ by taking $\El\!F$ to be the small set $F(V)$ if $F$ is
in $\Univ(V)$ and $\El\!\alpha(a,a')$ to be the set $\{0\,|\,\alpha a = a'\}$ if
$\alpha$ is an isomorphism between $F$ and $F'$ in $\Univ(V)$ and $a$ is in $F(V)$
and $a'$ is in $F'(V)$.

\begin{theorem}
The family $\UU\vdash \El\!$ is a {\em universal} small and discrete stack: If $\Gamma\vdash A$ is
small and discrete stack, there exists a {\em unique} map $|A|:\Gamma\rightarrow \UU$ such that
$\El |A| = A$ (with {\em strict} equality).
\end{theorem}

% We can then define $\Univ\vdash \El\!$ which is universal as small and discrete family.

\subsection{Dependent sums}

\begin{theorem}
If $\Gamma \vdash A$ and $\Gamma.A\vdash B$ have stack structures then we can glue descent data and paths in $\Gamma\vdash \Sigma A B$.
\end{theorem}

\begin{proof}
Let $((u_i, v_i),(\omega_{ij},\mu_{ij})) ∈ D(\Sigma AB)ρ$ be a descent datum on a covering $(U_i)_{i ∈ I}$
of $U$. We construct a glue $((u, v), (\omega_i,\mu_i))$ that commutes with restriction.

Let $(u,\omega_i)$ to be the glue of the datum $(u_i, \omega_{ij}) \in D(A)\rho$. We describe a descent datum in $D(B)(\rho,a)$. The object part of this descent datum is given by $v_i \inv{\omega_i} \in B(\rho|U_i, u|U_i)$. We have then paths $(v_i \inv{\omega_i}\uparrow \omega_i): v_i \inv{\omega_i} \pto_{\omega_i} v_i$
and thus paths
\begin{gather*}
(v_i \inv{\omega_i} \uparrow \omega_i)|U_{ij} \cdot \mu_{ij} \cdot \inv{(v_j \inv{\omega_j} \uparrow \omega_j)}|U_{ij}: v_i \inv{\omega_i}|U_{ij} \pto v_j \inv{\omega_j}|U_{ij}
\end{gather*}

These satisfy the cocycle condition. Thus we have a descent datum in $D(B)(\rho, u)$. Let $(v, \mu_i')$ be the glue of this datum.
We have paths
$\mu_i\coloneq \mu'_i \cdot (v_i\inv{\omega_i} \uparrow \omega_i) : v|U_i\pto_{\omega_i} v_i$

We then take the glue of $((u_i, v_i),(\omega_{ij},\mu_{ij}))$ to be given by $((u, v),(\omega_i,\mu_i))$. Since
\begin{gather*}
\mu_i'|U_{ij} \cdot (v_i \inv{\omega_i} \uparrow \omega_i)|U_{ij}\cdot \mu_{ij} \cdot \inv{(v_j \inv{\omega_j} \uparrow \omega_j)}|U_{ij} = \mu_j' |U_{ij}
\end{gather*}
we have that $\mu_i|U_{ij} \cdot \mu_{ij} = \mu_j|U_{ij}$.

Let $\alpha:\rho \pto \rho'$. Given a matching family of paths $(\omega_i, \mu_i) : (u, v) | U_i \pto_{\alpha|U_i} (u', v')$. Since $A$ have a stack structure we have a unique $\omega: u\pto_\alpha u'$ with $\omega|U_i = \omega_i$. But then $\mu_i:v|U_i \pto_{(\alpha, \omega)|U_i} v'|U_i$ is a matching family for the stack $B$ and thus have a unique $\mu : v \pto_{(\alpha, \omega)} v'$ where $\mu|U_i = \mu_i$.
\end{proof}

\subsection{Paths}

Descent data for the \emph{discrete} family $Γ ⊢ \Path{A}{a_0}{a_1}$
correspond to matching families of paths for $A$ and they have unique
glues if $A$ is a prestack. Unique choice then gives
us a function from descent data to glues for $\Path{A}{a_0}{a_1}$
which, also by uniqueness, necessarily commutes with restriction.

\begin{proposition}
If $\Gamma\vdash A$ has a stack structure and $\Gamma\vdash a_0:A$ and
$\Gamma\vdash a_1:A$ then $\Gamma\vdash \Path{A}{a_0}{a_1}$ has a
discrete stack structure.
\end{proposition}

%The requirement for the (set-valued) presheaf $Cρ(U) ≔
%\{ ω\,|\,ω : a_0(ρ|U) \pto a_1(ρ|U) \}$ to satisfy the sheaf property
%is exactly what determines the stack structure on the (discrete)
%family $Γ ⊢ C = \Path{A}{a_0}{a_1}$, where $Γ ⊢ a_0, a_1 : A$.

\subsection{Univalence}

An equivalence between two types $A$ and $B$ is a map $f: A\rightarrow B$ such that for each $y:B$ the fiber of $f$ above $y$ is contractible. If both $A$ and $B$ are discrete stacks, then $f$ being contractible means that $f$ is an isomorphism. As in \cite[5.4]{MR1686862}, we have then a one-to-one correspondence between the type of equivalences $\El a \simeq \El b$ and the type of paths $\Path{\Univ}{a}{b}$.

%The type
%%
%\begin{gather*}
%\Eq{A}{B} ≔ Σ(f : A → B)\isEq f
%\end{gather*}
%%
%is a subtype of $A → B$. For $A = \El a$ and $B = \El b$ we have that
%$\Eq{A}{B}$ is the discrete type of isomorphisms between $a$
%and $b$, which we may denote by $\Iso{a}{b}$.

%\begin{proposition}
%The map $\Iso{a}{b} → \Path{\Univ}{a}{b}$ is an inverse of $\tp{a}{b} : \Path{\Univ}{a}{b} → \Iso{a}{b}$.
%\end{proposition}

%\begin{proposition}
%In this model, $\isContr Σ(x : \Univ)\Eq{A}{(\El x)}$ is inhabited.
%\end{proposition}

\subsection{Propositional truncation}

 We define the family of sets $\Trunc{A}ρ$ inductively. For every
basic open $U$, $ρ ∈ Γ(U)$ and $u ∈ Aρ$ let $u ∈ \Trunc{A}ρ$.
Moreover, for every covering $(U_i)_{i ∈ I}$ of a given basic open
$U$, $ρ ∈ Γ(U)$ and $u_i ∈ \Trunc{A}(ρ|U_i)$ let $(U_i,u_i)_{i ∈ I} ∈
\Trunc{A}ρ$. Notice that this forms a set since the index set $I$ is restricted to be
small. (Without this restriction, we will get a {\em class} and not a {\em set} in general.)
Then, we define the family of functions $\Trunc{A}ρ → \Trunc{A}(ρ|V)$
recursively. For every pair of basic opens $V ⊆ U$, $ρ ∈ Γ(U)$ and $x
∈ \Trunc{A}ρ$ let $x|V ≔ u|V$ if $x = u$ with $u ∈ Aρ$ and $x|V ≔
(U_j∩V,u_j|U_j∩V)_{j ∈ J}$, where $J ⊆ I$ is the restriction to
indices $i ∈ I$ such that $U_i$ meets $V$, if $x = (U_i,u_i)_{i ∈
I}$ with $u_i ∈ \Trunc{A}(ρ|U_i)$. Lastly, we define the type family $\Trunc{A}$ to be the
proposition on the family of sets and functions just defined. The collection of discrete families $Γ ⊢
A\ \textsf{discrete}$ is \emph{not} closed under propositional
truncation: Given $α : ρ \pto ρ'$ and $x ∈ \Trunc{A}ρ$, then there is
a unique path $x \pto_α x'$ that connects $x$ to each element $x' ∈
\Trunc{A}ρ'$.

The family $Γ ⊢ \Trunc{A}$ always has a stack structure, even without assuming one on $Γ ⊢ A$.
If we have a covering $(V_l)_{l\in L}$ of $U$ and for each $l$ in $L$ we have an element
$x_l$ of $\norm{A}(\rho|V_l)$, then this family $x_l$
always defines in a unique way a descent datum and we can consider the family
$(V_l,x_l)_{l\in L}$, which defines a gluing
of the family $x_l$. This operation furthermore satisfies the functoriality condition.

\subsection{Example: One-point space}

One of the simplest examples of the notion of stack is that where the poset of basic opens has exactly one object $V$. In that case, a covering of $V$ is a nonempty finite family $(V_i)_{i\in I}$ where each $V_i = V$ and a stack is a single groupoid $G$ with a gluing operation. That is, for any family $(u_i)_{i\in I}$ of elements in $G$ and paths $\varphi_{ij}: u_i \pto u_j$ satisfying $\varphi_{ij} · \varphi_{jk} = \varphi_{ik}$ we have $\glue(u_i, \varphi_{ij}) = (u, \varphi_i)$ such that $\varphi_i:u\pto u_i$ and $\varphi_i · \varphi_{ij} = \varphi_j$. We can use this example to motivate the definition of propositional truncation given above. Suppose we naively truncate the groupoid $G$ to get the proposition $\Trunc{G}$ on the objects of $G$, then we have no way of defining the gluing operation on $\Trunc{G}$ since we lack a particular choice of glue for a given descent datum.

\onlyinsubfile{
  \bibliography{bibliography}
}

\section{Countable choice}

\label{sec:choice}
\onlyinsubfile{
  \title{Independence of countable choice from type theory}
  \maketitle
}

\subsection{A stack model where countable choice does not hold}

 We write $U,V,W,\dots$
nonempty open rational intervals included in the open unit interval $(0,1)$. For each $n$, and $i = 1,\dots,n$
we let $U^n_i$ be $((i-1)/(n+1),(i+1)/(n+1))$
so that $(U^n_i)_{i = 1,\dots,n}$ is a covering of $(0,1)$.

 We let $|\NN|$ be the constant presheaf where each $|\NN|(V)$ is the set $\nats$ of natural numbers
and $\NN = \El |\NN|$. We have \cite{vandalen.ii}
\begin{lemma}
$|\NN|$ is a (small) sheaf.
\end{lemma}

 It is also well-known that in the sheaf model over $(0,1)$, there are Dedekind reals
that are not Cauchy reals \cite{vandalen.ii}. It is simple to transform this fact to a counter-example
to our type-theoretic version of countable choice.

We define
$A:\NN\rightarrow \UU$ by letting $A\, n$ be the subsheaf of the (small) constant sheaf $|\mathsf{Q}|(V) = \rats$ of rational numbers
\begin{gather*}
(A\, n)(V) = \left\{r\in\rats\, \middle|\, \forall (x\in V)\, \abs{x-r}<\frac{1}{n+1}\right\}
\end{gather*}

 Notice that each ${i}/(n+1)$ is an element of $(A\, n)(U^n_i)$.

\begin{proposition} In this model
\begin{enumerate}
\item the type $\Pi (n:\NN)\norm{\El (A\, n)}$ is inhabited
\item the type $\Pi (n:\NN)\El (A\, n)$, and hence also the type $\norm{\Pi (n:\NN)\El (A\, n)}$, is empty
\end{enumerate}
\end{proposition}

\begin{proof}
 For each open set $V$, we let
$s_V\, n$ be the family $(V\cap U^n_i,i/(n+1))$, $i$ such that $V$ and $U^n_i$ meet,
in $\norm{\El (A\, n)}(V)$. Since we have
$(s_V\, n)|W = s_W\, n$ if $W\subseteq V$, this defines a section of  $\Pi (n:\NN)\norm{\El (A\, n)}$.

  For the second point, it is enough to notice that, for each given $V$, the set
\begin{gather*}
(A\, n)(V) =  \left\{r\in\rats\, \middle|\, \forall (x\in V)\, \abs{x-r}<\frac{1}{n+1}\right\}
\end{gather*}
is empty for $n$ large enough.
\end{proof}

\begin{corollary}
In this model, the principle of countable choice $\CC$ does not hold.
\end{corollary}

\begin{corollary}
One cannot show countable choice in type theory with one univalent universe and propositional truncation.
\end{corollary}

\onlyinsubfile{
  \bibliography{bibliography}
}

\section{Markov's principle}

\label{sec:markov}
\onlyinsubfile{
  \title{Independence of Markov's principle from Type Theory}
  \maketitle
}

 The interpretation of the type $\NN$ was especially simple on the space $(0,1)$
using the fact that its basic opens are {\em connected}. We will now consider the ``dual''
case where the space is {\em totally disconnected}. We assume from now on that
the basic opens are nonzero elements $e,e',\dots$ of a Boolean algebra with decidable equality.
We consider only coverings of $e$ given by a finite partition $e_i$, $i\in I$, of $e$,
that is a
{\em finite set} of disjoint elements $e_i \leqslant e$ such that
$e = \bigvee_{i\in I} e_i$.

Given a type family $\Gamma\vdash A$ and $\rho \in \Gamma(e)$, a descent datum $d ∈ D(A)\rho$ for this family
is now simply given by a partition $e_i$, $i\in I$, of $e$ and a family $u_i \in A\rho | e_i$.

We can now {\em strengthen} the notion of stack structure by further imposing that
we have $(\glue d)|e_i = u_i$ for $d=(u_i) \in D(A)\rho$. This \emph{strict gluing condition} states that the required equalities between $(\glue d)|e_i$ and $u_i$
are {\em strict} equalities.
In fact, it is enough to require $\glue (u) = u$ for partitions consisting of exactly one element.

This refinement is needed for the elimination of natural numbers and Booleans in the universe.

\begin{proposition}
If $\Gamma.A\vdash B$ satisfies the strict gluing condition, then so does $\Gamma\vdash\Pi A B$.
\end{proposition}

\begin{proposition}
If $\Gamma\vdash A$ and $\Gamma.A\vdash B$ satisfy the strict gluing condition, then so
does $\Gamma\vdash\Sigma A B$.
\end{proposition}

 If $\Univ(e)$ is the collection of sheaves on $e$, we can refine the stack structure on $\Univ$
in order to satisfy the strict gluing condition: If $e_i$ is a partition of $e$ and $F_i$ is a
sheaf on $e_i$ we define $F = \glue (e_i,F_i)$ by taking $F(e')$, for $e'\leqslant e$,
to be the product of all $F_i(e'\wedge e_i)$ if $e'$ meets strictly more than one $e_i$, and to be
{\em exactly} $F_i(e')$ if $e'\leqslant e_i$. This defines a sheaf, and the functoriality law
$\glue (e_i,F_i)|e' = \glue (e_i\wedge e',F_i|e_i\wedge e')$ is satisfied.

\subsection{Natural numbers and Booleans}

 We define the sheaf $|\Nat|$ by taking $|\Nat|(e)$ to be the set of families $(e_i,n_i)$ where $e_i$
is a partition of $e$ and $n_i\neq n_j$ if $i\neq j$.
We define similarly $|\Bool|$ where $n_i$ can only take the values $0$ or $1$, and
$|\Top|(e) = \{0\}$, and $|\Bot|(e)$ is the empty set.
We define then $\NN = \El |\NN|$ and similarly for $\Bool$, $\Top$ and $\Bot$.

We define $\suc (e_i,n_i)$ to be $(e_i,n_i+1)$ and $\0(e)$ is the element $(e,0)$.

The $\natrec\!$ operator is then defined as a section of $Γ.\Nat ⊢ C$
\begin{alignat*}{2}
& (\natrec c\, d)(\rho,0)         &&= c\rho \\
& (\natrec c\, d)(\rho,n+1)       &&= d(\rho, n, (\natrec c\, d)(\rho,n)), \text{ where } n\in \nats\\
& (\natrec c\, d)(\rho,(e_i,n_i)) &&= \glue (e_i,(\natrec c \, d)(\rho|e_i,n_i))
\end{alignat*}
given sections $Γ ⊢ c : C\ir{\0}$ and $Γ.\Nat.C ⊢ d : C\ir{\suc \q}\p$.

We remark that the strict gluing condition is needed to make the above definition work, i.e.~so that for $m \in \Nat(e)$ and $e' ⩽ e$ we have $((\natrec c\, d)\rho \, m)|e' = (\natrec c\, d)\rho|e' \, m|e'$.

\subsection{A stack model where Markov's principle does not hold}

We can express Markov's principle in type theory by the type:
\begin{gather*}
\MP \coloneq \Pi(h: \Nat\rightarrow \Bool)(\neg\neg (\Sigma (x: \Nat)\El \isZero\,(h\,x)) \rightarrow \Sigma (x: \Nat) \El \isZero\,(h\,x))
\end{gather*}
where $\isZero: \Bool \rightarrow \Univ$ is defined by $\isZero \coloneq \lambda y.\boolrec \In{\Top}\,\In{\Bot}\,y$ and the type $¬A$ by $A → \Bot$.

We could also consider the version where we use weak existential $\exists (x:A) B = \norm{\Sigma (x:A) B}$
instead of sigma type, but the two versions are logically equivalent \cite[Exercise 3.19]{hottbook}.

 Take a countably infinite set of variables $p_0,p_1,\dots$.
Consider the free Boolean algebra generated by the atomic formulae $p_n$.
We write $p_n = 0$ for $\neg p_n$ and $p_n = 1$ for $p_n$.
An object $e$ in this algebra represents then a compact open in Cantor space $\{0,1\}^\nats$, where
a conjunctive formula $\bigwedge p_i = b_i$ represents the set of sequences in $\{0,1\}^\nats$
having value $b_i$ at index $i$. A formula $e$ in the algebra is then a finite disjunction of these.

 We have an interpretation of type theory in stacks over this algebra, and we are going to see that
Markov's principle is not valid in this interpretation. We define $\f$ in $\Nat \rightarrow \Bool$ by
taking $\f n$, $n ∈ \Nat(e)$, at $e_i$ to be $((e_0,0),(e_1,1))$ where $e_b$ is $e_i\wedge (p_{n_i} = b)$ if $e_i$ meets both
$(p_{n_i} = 0)$ and $(p_{n_i} = 1)$, and to be $(e_i,b)$ if $e_i\leqslant (p_{n_i} = b)$.

\begin{proposition}
In this model
\begin{enumerate}
\item $\neg\neg (\Sigma (x: \Nat)\,\El \isZero\,(\f x))$ is inhabited.
\item $\Sigma (x: \Nat)\, \El \isZero\,(\f x)$ is not inhabited.
\end{enumerate}
\end{proposition}
\begin{proof}
To show that $\neg\neg (\Sigma (x: \Nat) \El \isZero\,(\f x))$ is inhabited it is sufficient
to show that for all $e$ the set $(\neg (\Sigma (x: \Nat) \El \isZero\,(\f x)))(e)$ is empty.
For that it will be sufficient to show that for some $e' \leqslant e$
we have that $(\Sigma (x: \Nat) \El \isZero\,(\f x))(e')$ is not empty.
But given any $e$ we can simply choose $e' = (p_n = 0) \land e$ for some $n$ big enough.
Thus $\El \isZero\,(\f n)$ at $e'$ is $\{0\}$ and $(\Sigma (x: \Nat) \El \isZero\,(\f x))(e')$ is not empty.

We now show that $\Sigma (x: \Nat) \El \isZero\,(\f x)$ is not inhabited. For any $n = (e_i,n_i)$  in $\Nat(1)$
where $(e_i)$ is a partition of $1$, we can find exactly one $e_i$ which contains (as a compact open subset
of Cantor space) the constant function $1$. This element $e_i$ meets $p_{n_i} = 1$ so that
$\El \isZero\,(\f n)$ is the empty set at $(p_{n_i} = 1) \land e_i$ and hence also at $1$.
%% we have for $m$ big enough
%% %
%% \begin{gather*}
%% (p_0 =1) \land \dots \land (p_m = 1) \leqslant e_i
%% \end{gather*}
%% %
%% If $m$ is bigger than $n_i$ we see that
\end{proof}

\begin{corollary}
In this model Markov's principle does not hold.
\end{corollary}

\begin{corollary}
One cannot show Markov's principle in type theory with one univalent universe.
\end{corollary}

 The situation however is different from the one of countable choice. The following
provides an alternative argument that Markov's principle cannot be proved
in type theory with one univalent universe\footnote{This argument gives also a proof that Markov's principle is independent of a hierarchy
of univalent universes by considering the cubical set model \cite{CCHM} in a set theory
where Markov's principle does not hold.}.

\begin{proposition}
Markov's principle does not hold in the groupoid model in a set theory where Markov's principle
does not hold (for instance in suitable sheaf models of CZF \cite{GAMBINO2006164}).
\end{proposition}

In \cite{CM} it was shown that Markov's principle is independent from type theory with one (non-univalent) universe. The paper describes an extension of type theory where the principle does not hold and proves the consistency of that extension with a normalization argument. We note however that the model given here does not give an interpretation of the extended type theory in \cite{CM}. In particular the universe (inductively defined) in that extension satisfies the sheaf property.

\onlyinsubfile{
  \bibliography{bibliography}
}

\section{Conclusion}

 One special case of sheaf models are Boolean-valued models, for instance as in the work
\cite{MR3227408}, and it would be interesting to formulate a stack version of these models as well.

 We expect that essentially the same kind of models can be defined over a {\em site} and not only
over a topological space. In particular, it should be possible to extend the sheaf model in \cite{MC}
to a stack model of type theory with an algebraic closure of a given field, where existence
of roots is formulated using propositional truncation (as explained in the cited work, this existence
cannot be stated using strong existence expressed by sigma types). Another example could be a stack version
of Schanuel topos used in the theory of nominal sets \cite{Pitts}.

 As  stated in the introduction, the argument should generalize to an $\infty$-stack version
of the cubical set model \cite{CCHM}. The coherence condition on descent data will be
infinitary in general, but it will become finitary when we restrict the homotopy level (and empty in particular
in the case of propositions). 
%% It is delicate to find the correct strict
%% equality laws in order to obtain a model of dependent type theory. However, we believe that these
%% laws will be exactly functoriality of gluing and the composition operations of \cite{CCHM},
%% which generalize the path lifting operation.

\bibliography{bibliography}

\end{document}